\documentclass{article}
\usepackage[a4paper,margin=1in]{geometry}
\usepackage[T1]{fontenc}
\usepackage{setspace}
\usepackage{graphicx} % Required for inserting images
\usepackage{amsmath}
\usepackage{mathtools}
\usepackage{amsthm}
\usepackage{tablefootnote}
\usepackage{braket}
\usepackage{amssymb}
\usepackage{mathrsfs}
\usepackage{amsfonts}
\usepackage{authblk}
\usepackage{makecell}
\usepackage{threeparttable}
\usepackage{bbding}
\usepackage{subfig}
\usepackage{dcolumn}
\usepackage{float}
\usepackage{comment}
\usepackage{multirow}
\usepackage{color}
\usepackage{tcolorbox}
\usepackage{algorithm}
\usepackage{algorithmic}
\usepackage{tikz}
\usepackage{appendix}
\usepackage[linktocpage=true, colorlinks,
linkcolor=blue,citecolor=blue,
bookmarks,bookmarksopen,bookmarksnumbered]
{hyperref}
\usepackage{circuitikz} 
\usepackage{tikz-3dplot} 
\usepackage{adjustbox}
\allowdisplaybreaks[4]

\newcommand{\rbra}[1]{\left( #1 \right)} 
\newcommand{\sbra}[1]{\left[ #1 \right]}
\newcommand{\cbra}[1]{\left\{ #1 \right\}}
\newcommand{\abs}[1]{\lvert #1 \rvert}
\newcommand{\Abs}[1]{\lVert #1 \rVert}

\newcommand{\tr} {\operatorname{tr}}
\newcommand{\diag} {\operatorname{diag}}
\theoremstyle{plain}

\newtheorem{theorem}{Theorem}[section]
\newtheorem{proposition}[theorem]{Proposition}
\newtheorem{lemma}[theorem]{Lemma}
\newtheorem{corollary}[theorem]{Corollary}
\theoremstyle{definition}
\newtheorem{definition}[theorem]{Definition}

\theoremstyle{remark}

\begin{document}
\title{Quantum Approximate Counting with Bernoulli Oracles}
\author[1]{Chengshen Gao}
\author[2]{Yongzhen Xu\thanks{Corresponding author: xuyongzh@gmail.com}}
\author[1,2]{Lvzhou Li\thanks{Corresponding author: lilvzh@mail.sysu.edu.cn}}
\affil[1]{Institute of Quantum Computing and Software, School of Computer Science and Engineering, Sun Yat-sen University, Guangzhou 510006, China}
\affil[2]{Quantum Science Center of Guangdong-Hong Kong-Macao Greater Bay Area, Shenzhen, 518045, Guangdong, China}
\date{}

\maketitle

\begin{abstract}
Quantum counting is a fundamental quantum algorithm that estimates the fraction of marked
elements using a membership oracle, achieving a quadratic speedup over classical sampling.
The membership oracle, however, assumes exact labeling of each element, but this assumption fails when the labels are inherently probabilistic.
We study quantum counting with \emph{Bernoulli oracles}, where given $m$ Bernoulli distributions
with unknown biases $p_1,\dots,p_m$ and a gap parameter $\Delta$, the goal is to estimate the fraction
$\rho$ of \emph{positive} distributions ($p_i\ge1/2+\Delta$) to within additive error
$\epsilon$. We prove an upper bound of
$\tilde{O}\!\big(\frac{\sqrt{\rho}}{\Delta\epsilon}+\frac{1}{\Delta\sqrt{\epsilon}}\big)$
queries, achieving a quadratic speedup over the classical sample complexity. % of $\Theta(\rho/\Delta^2\epsilon^2)$. 
  Our algorithm first uses the Quantum Singular Value
Transformation (QSVT) to coherently amplify the bias gap without collapsing the superposition over
distributions, and then applies two-stage adaptive amplitude estimation. We complement this upper bound with
a near-matching lower bound of $\Omega(\sqrt{\rho}/(\Delta\epsilon))$ via a new composition
theorem for the quantum adversary method in the Boolean-over-average-case direction. For the special case of a constant gap $\Delta=\Theta(1)$, which corresponds to the
bounded-error oracle where each query returns the correct label with constant probability, our
bounds specialize to $\tilde{O}\big(\frac{\sqrt{\rho}}{\epsilon}+\frac{1}{\sqrt{\epsilon}}\big)$
and $\Omega(\frac{\sqrt{\rho}}{\epsilon})$, thereby  characterizing the query complexity of
quantum counting with bounded-error oracles. 
\end{abstract}

\section{Introduction}
One of the main motivations for works in quantum computing is the prospect of fast quantum algorithms for tackling fundamental computational problems~\cite{montanaro2016quantum,zhang2022brief}. Among the most prominent examples is quantum counting~\cite{10.5555/646252.686013}, which addresses the following task: given a membership oracle for an $m$-element set that exactly labels each element as either $0$ or $1$, estimate the fraction $\rho$ of elements labeled $1$ to within additive error $\epsilon$. Quantum counting solves this task with $\Theta\!\big(\frac{\sqrt{\rho}}{\epsilon}\big)$ queries, a quadratic improvement over classical random sampling, which requires $\Theta\!\big(\frac{\rho}{\epsilon^2}\big)$ samples. This quadratic speedup has been a recurring theme in quantum algorithm design: amplitude estimation~\cite{brassard2002quantum}, quantum mean estimation~\cite{hamoudi2021quantum,kothari2023mean}, quantum distribution testing~\cite{gilyen2020distributional,luo2024succinct}, and quantum machine learning~\cite{pmlr-v267-luo25e,zhang2024quantum}, all of which stem from the same underlying principle and achieve similar speedups over their classical counterparts.

The membership oracle, however, is an idealized abstraction: it assumes that whether an element is marked can be determined with certainty from a single query. In many realistic settings, this labeling is itself subject to statistical uncertainty. For instance, in crowdsourcing~\cite{kittur2008crowdsourcing}, whether an annotator is reliable is not a deterministic label but a latent bias parameter governing the probability of correct annotation. In multi-armed bandits~\cite{slivkins2019introduction}, each arm's optimality depends on the mean of an unknown reward distribution, and in empirical Bayes methods~\cite{robbins1951asymptotically,robbins1992empirical}, parallel Bernoulli observations with heterogeneous biases are aggregated to infer the prior. In such scenarios, each query returns not a definite 0/1 label, but a sample from a \emph{Bernoulli distribution} whose bias parameter encodes the property of interest. This naturally leads to the problem of quantum counting with the Bernoulli oracle, which is formalized as follows.

We are given a collection of $m$ Bernoulli distributions $\mathcal{D} =\{D_1,\ldots,D_m\}$ on $\cbra{0,1}$ with unknown parameters $p_1,\dots,p_m \in [0,1]$, where $\Pr[D_i=1] = p_i$. Each distribution is either \emph{positive} ($p_i \ge 1/2 + \Delta$) or \emph{negative} ($p_i \le 1/2 - \Delta$), where $\Delta\in \left(0,\frac{1}{2}\right]$.  The task is now  to output an estimate $\hat{\rho}$ of the true fraction
$\rho = \frac{1}{m} \lvert\{i : p_i \ge 1/2 + \Delta\}\rvert$
such that $\lvert\hat{\rho} - \rho\rvert \le \epsilon$.  In the classical setting, one may adaptively select any distribution and draw samples from it. This problem has been studied by Lee and Valiant~\cite{lee2021uncertainty}, who established a tight sample complexity of $\Theta\!\left(\frac{\rho}{\Delta^{2}\epsilon^2} \log\frac{1}{\delta}\right)$ for success probability $1-\delta$.
The $1/\Delta^2$ and $1/\epsilon^2$ scaling means that, when the bias gap $\Delta$ is small or a high-precision estimate is required,
the sample size quickly becomes prohibitive.

In the quantum setting, we are given query access to a quantum oracle $O$ that acts on an index register and two auxiliary
registers as
\begin{align}\label{eq: input oracle}
    O\ket{i}_A\ket{\vec{0}}_{B,C}=\ket{i}_A\left(\sqrt{p_i}\ket{1}_B\ket{\varphi_{i,1}}_C+\sqrt{1-p_i}\ket{0}_B\ket{\varphi_{i,0}}_C\right),
\end{align}
where $i\in\cbra{1,...,m}$. %The goal is to output an estimate $\hat{\rho}$ of the true fraction $\rho = \frac{1}{m} \lvert\{i : p_i \ge 1/2 + \Delta\}\rvert$ such that $\lvert\hat{\rho} - \rho\rvert \le \epsilon$, using as few queries to $O$ as possible.

\begin{comment}
The most prominent among these is the bounded-error oracle, which models the following setting: for each index $i\in\sbra{N}$ there is a true label $x_i \in \{0,1\}$, and a  query to $i$ returns a bit $\tilde{x}_i$ satisfying
$\Pr[\tilde{x}_i = x_i] \ge \frac{9}{10}$. 
The Bernoulli oracle subsumes the membership oracle and the bounded-error
oracle as special cases.
When $\Delta=1/2$, each $p_i$ is either $0$ or $1$, so every query
exactly reveals whether the distribution is positive or negative,
which is precisely the membership oracle.
When $\Delta=\Theta(1)$, interpreting the outcome $\ket{0}$ as
'positive' and $\ket{1}$ as 'negative' yields a procedure that
identifies the correct label with probability at least $1/2+\Delta$,
which is precisely the bounded-error oracle.
\end{comment}

The idea of extending quantum algorithms to oracles beyond the exact membership model is not without precedent. Quantum search has been studied with the bounded-error oracle~\cite{hoyer2003quantum}, the noisy oracle~\cite{rosmanis2023quantum}, and the neutral oracle~\cite{kociumaka2025near}. Quantum $k$-minimum finding has been extended to the untrustworthy oracle~\cite{quek2020robust}, the bounded-error oracle~\cite{wangtcs2024quantum}, and the approximate oracle~\cite{gao_et_al:LIPIcs.ESA.2025.51}. 

Meanwhile, a similar form of quantum query access~\eqref{eq: input oracle} has been applied to several problems over Bernoulli distributions. These include estimating the mean of a single distribution~\cite{10.5555/646252.686013,hamoudi2021quantum,kothari2023mean,nayak1999lower}, testing properties between two distributions~\cite{belovs2019quantum,gilyen2020distributional,luo2024succinct,canonne_et_al:LIPIcs.TQC.2025.7}, and distinguishing distributions of arbitrarily close bias~\cite{aaronson2011advice}. All of these results, however, concern properties of individual distributions or pairwise comparisons. By contrast, quantum counting with the Bernoulli oracle concerns a global quantity: the fraction $\rho$ of distributions whose bias exceeds $1/2+\Delta$. 

With this quantum oracle model in place, the natural question raises: 
\begin{center}
    \emph{Can quantum counting with the Bernoulli
oracle outperform classical sampling}?
\end{center}

Here, we answer this question in the affirmative: quantum counting generalizes to the Bernoulli oracle while preserving its quadratic speedup, up to a logarithmic factor. A near-matching lower bound confirms that this speedup is essentially optimal. 

\subsection{Definitions of the Input Oracles}

In this section, we define the three input oracles relevant to this work and
establish the containment chain among them, which transfers query-complexity
bounds between the corresponding counting problems.

\begin{enumerate}
    \item \textbf{Membership oracle.}
    Let $f\colon [m]\to\{0,1\}$ be a function. The algorithm has query
    access to a unitary $O_f$ acting on an index register $A$ and an answer
    register $B$ as
    \begin{align}
        O_f \ket{i}_A \ket{0}_B = \ket{i}_A \ket{f(i)}_B.
    \end{align}

    \item \textbf{Bounded-error oracle.}
    Let $f\colon [m]\to\{0,1\}$ be a function. The algorithm has query 
    access to a unitary $O_f^e$ acting on an index register $A$ and an answer
    register $B$ as
    \begin{align}
        O_f^e \ket{i}_A \ket{0}_B=\ket{i}_A \bigl(\sqrt{p_i^e} \ket{f(i)}_B+ \sqrt{1-p_i^e}\ket{\bar{f(i)}}_B\bigr),
    \end{align}
    where $p_i^e \ge \frac{9}{10}$ denotes the probability that a single
query to index $i$ returns the correct label~\cite{hoyer2003quantum}.

    \item \textbf{Bernoulli oracle.}
    As introduced above, we are given $m$ Bernoulli distributions $\mathcal{D} = \{D_1, \dots, D_m\}$ on $\{0,1\}$ with unknown parameters
    $p_1, \dots, p_m \in [0,1]$, where $\Pr[D_i = 1] = p_i$. Each
    distribution is either \emph{positive} ($p_i \ge 1/2 + \Delta$) or
    \emph{negative} ($p_i \le 1/2 - \Delta$), with
    $\Delta \in (0, \frac{1}{2}]$. The algorithm has query access to a
    unitary $O$ acting as
    \begin{align}
        O \ket{i}_A \ket{\vec{0}}_{B,C} = \ket{i}_A \bigl( \sqrt{p_i} \ket{1}_B \ket{\varphi_{i,1}}_C + \sqrt{1-p_i} \ket{0}_B \ket{\varphi_{i,0}}_C \bigr).
    \end{align}
\end{enumerate}

The three oracles are nested. Setting $p_i^e=1$ in the bounded-error oracle
recovers the membership oracle. Moreover, a single query to the
bounded-error oracle returns a sample of a Bernoulli distribution: it
returns $f(i)$ with probability $p_i^e\ge\frac{9}{10}$ and
$\bar{f(i)}$ otherwise, so the bias of the returned bit lies at distance
at least $\frac{2}{5}$ from $\frac{1}{2}$. Hence
\begin{align}\label{eq:oracle-hierarchy}
    \text{membership}\;\subset\;\text{bounded-error}\;\subset\;\text{Bernoulli},
\end{align}
where membership corresponds to the Bernoulli oracle with gap
$\Delta=\frac{1}{2}$ (equivalently, $p_i\in\{0,1\}$), and bounded-error
to gap $\Delta=\frac{2}{5}$.

The following proposition makes this nesting precise at the level of the
counting problems, showing in particular that bounded-error counting
reduces to counting positive distributions.

\begin{proposition}[Reduction from bounded-error counting to
counting positive distributions]\label{prop:be-to-bernoulli}
Let $O_f^e$ be a bounded-error oracle for an unknown function
$f\colon[m]\to\{0,1\}$ with correctness probabilities
$p_i^e\ge\frac{9}{10}$ for all $i\in[m]$, let
$\rho=\frac{1}{m}\bigl|\cbra{i\in[m]:f(i)=1}\bigr|$ be the fraction of
ones of $f$, and for each $i\in[m]$ let $D_i$ be the distribution of the
bit returned by a single query to index $i$. Then
$\mathcal{D}=\{D_1,\dots,D_m\}$ is a valid Bernoulli counting instance
with gap $\Delta=\frac{2}{5}$, in which $D_i$ is positive if and only if
$f(i)=1$. Moreover, each query to the Bernoulli oracle of $\mathcal{D}$
is simulated by one query to $O_f^e$. Consequently, any quantum algorithm that
estimates the fraction of positive distributions to within additive
error $\epsilon$ with success probability $1-\delta$ using $q$ queries
estimates $\rho$ with the same error and success probability, using $q$
queries to $O_f^e$.
\end{proposition}

\begin{proof}
A query to index $i$ returns $f(i)$ with probability
$p_i^e\ge\frac{9}{10}$ and $\bar{f(i)}$ with the remaining probability,
so $D_i$ is a Bernoulli distribution with parameter
\begin{equation}\label{eq:induced-bias}
p_i=\begin{cases}
p_i^e, & f(i)=1,\\
1-p_i^e, & f(i)=0.
\end{cases}
\end{equation}
Thus $p_i\ge\frac{9}{10}=\frac{1}{2}+\frac{2}{5}$ when $f(i)=1$, and
$p_i\le\frac{1}{10}=\frac{1}{2}-\frac{2}{5}$ when $f(i)=0$. Hence
$\mathcal{D}$ is a valid Bernoulli counting instance with gap
$\Delta=\frac{2}{5}$ and $D_i$ is positive exactly when $f(i)=1$, so
that
\begin{equation}\label{eq:answer-preserved}
\frac{1}{m}\,\bigl|\cbra{i : D_i \ \text{is positive}}\bigr|
\;=\;
\frac{1}{m}\,\bigl|\cbra{i : f(i)=1}\bigr|
\;=\; \rho .
\end{equation}

For the simulation, apply $O_f^e$ with the answer register $B$
initialized to $\ket{0}$, append an auxiliary qubit $C$ initialized to
$\ket{0}$, and apply a CNOT from $B$ to $C$. The resulting state is
\[
\ket{i}_A\Bigl(\sqrt{p_i}\,\ket{1}_B\ket{1}_C
+\sqrt{1-p_i}\,\ket{0}_B\ket{0}_C\Bigr),
\]
which is of the form of Eq.~\eqref{eq: input oracle} with
$\ket{\varphi_{i,1}}=\ket{1}$ and $\ket{\varphi_{i,0}}=\ket{0}$. The
simulation costs one query to $O_f^e$ plus one CNOT gate and one
auxiliary qubit, which yields the last claim.
\end{proof}

The inclusion chain~\eqref{eq:oracle-hierarchy}, made precise by
Proposition~\ref{prop:be-to-bernoulli}, translates into relations between
query complexities. Upper bounds transfer downward: an upper bound for
the Bernoulli oracle with gap $\Delta$ applies to the bounded-error
oracle (via $\Delta=\frac{2}{5}=\Theta(1)$) and to the membership oracle
(via $\Delta=\frac{1}{2}$). Lower bounds transfer upward: the
$\Omega\bigl(\frac{\sqrt{\rho}}{\epsilon}\bigr)$ lower bound for
membership counting~\cite{10.5555/646252.686013,nayak1999lower} implies
the bound for bounded-error counting, which is matched by our
$\Omega\bigl(\frac{\sqrt{\rho}}{\Delta\epsilon}\bigr)$ lower bound at
$\Delta=\Theta(1)$. Table~\ref{tab:oracle} summarizes these
specializations.
\subsection{Our Results}

We now state our main results, which together give an essentially complete
characterization of quantum counting with the Bernoulli oracle.
Table~\ref{tab:oracle} summarizes the query complexity under each access model,
with our bounds highlighted.

\begin{table}[htbp]
\centering
\renewcommand{\arraystretch}{2.2}
\caption{Quantum counting under different access models.}
\begin{tabular}{|l|c|c|c|}
\hline
\textbf{Access model} & \textbf{Gap $\Delta$} & \textbf{Upper bound} & \textbf{Lower bound}\\[4pt]
\hline
Classical sampling \cite{lee2021uncertainty} 
& $(0,1/2]$
& \multicolumn{2}{c|}{$\Theta(\frac{\rho}{\Delta^2\epsilon^2})$}\\[4pt]
\hline
Membership \cite{10.5555/646252.686013,nayak1999lower}
& $1/2$
&\multicolumn{2}{c|}{$\Theta(\frac{\sqrt{\rho}}{\epsilon})$}\\[4pt]
\hline
Bounded-error\textbf{(this work)}
& $\Theta(1)$
& $\tilde{O}(\frac{\sqrt{\rho}}{\epsilon}+\frac{1}{\sqrt{\epsilon}})$
& $\Omega(\frac{\sqrt{\rho}}{\epsilon})$\\[4pt]
\hline
Bernoulli\textbf{(this work)}
& $(0,1/2]$
& $\tilde{O}(\frac{\sqrt{\rho}}{\Delta\epsilon}+\frac{1}{\Delta\sqrt{\epsilon}})$
& $\Omega(\frac{\sqrt{\rho}}{\Delta\epsilon})$\\
\hline
\end{tabular}
\label{tab:oracle}
\end{table}

\begin{theorem}[Upper bound, informal version of Theorem~\ref{the: Adaptive fraction estimator}]
    There exists a quantum algorithm that, given quantum query access to the Bernoulli oracle with bias $\Delta$, estimates the fraction $\rho$ of positive distributions to within additive error $\epsilon$ with success probability at least $1-\delta$, using $\tilde{O}\!\left(\left(\frac{\sqrt{\rho}}{\Delta\epsilon}+\frac{1}{\Delta\sqrt{\epsilon}}\right)\log\frac{1}{\delta}\right)$ queries.
\end{theorem}

\begin{theorem}[Lower bound, informal version of Corollary~\ref{cor:bernoulli-lower}]
    Any quantum algorithm that estimates $\rho$ to within additive error $\epsilon$ with success probability at least $2/3$, given quantum query access to a Bernoulli oracle with bias $\Delta$,  requires $\Omega\!\left(\frac{\sqrt{\rho}}{\Delta\epsilon}\right)$ queries.
\end{theorem}

The near-matching lower bound shows that the upper bound is tight up to logarithmic factors across the full range of $\Delta$. As discussed above, the bounded-error oracle corresponds to
$\Delta = \Theta(1)$. Hence our general bounds immediately specialize to the
following.

\begin{corollary}[Quantum counting with the bounded-error oracle]\label{cor:bounded-error}
For the bounded-error oracle~\cite{hoyer2003quantum}, which corresponds to the
special case $\Delta=\Theta(1)$, the above bounds specialize to an upper bound of
$\tilde{O}\!\big(\frac{\sqrt{\rho}}{\epsilon}+\frac{1}{\sqrt{\epsilon}}\big)$
and a lower bound of $\Omega(\frac{\sqrt{\rho}}{\epsilon})$.
\end{corollary}

\subsection{Techniques}
\paragraph{Upper bound.}

A direct approach is to first determine whether each distribution is positive or
negative, and then supply these binary labels to a quantum amplitude estimation
procedure to recover the fraction $\rho$. The difficulty is that distinguishing
positive from negative distributions cannot, in general, be performed
perfectly with a bounded number of queries: any decision subroutine errs with some nonzero probability, and this error propagates through the outer estimation,
compromising the accuracy of the final estimate.

The standard technique for handling imperfect decisions is to repeat the subroutine $O(\log(1/\delta))$ times and take the median, thereby reducing the error probability to $\delta$. This repetition-and-median approach, however, requires \emph{measuring} each distribution's output individually. Each measurement collapses the superposition over the index register, destroying the coherence that amplitude estimation relies on. Classical error reduction and quantum amplitude estimation are thus incompatible when composed directly.

We resolve this tension by performing error reduction \emph{coherently}.
Recall that each query to the oracle~\eqref{eq: input oracle} produces a state whose $B$-register amplitudes encode the bias $p_i$ of the queried distribution.
Ideally, we would like to apply a transformation to this state whose effect is a step function of the bias: if $p_i \ge 1/2+\Delta$, the $B$~register becomes
$\ket{0}$ with certainty; if $p_i \le 1/2-\Delta$, it becomes $\ket{1}$ with
certainty. This would eliminate the decision error entirely while preserving the
superposition over the index register. A perfect step function, however, is not a
polynomial of the amplitudes and therefore cannot be realized as a finite quantum
circuit.

The Quantum Singular Value Transformation (QSVT)~\cite{gilyen2019quantum} applies
polynomial transformations to the singular values of a block-encoded operator.
For the oracle~\eqref{eq: input oracle}, the amplitudes can be arranged as singular values via a suitable block encoding, so that a polynomial applied to the singular values translates to a polynomial transformation of the amplitudes. We thus approximate the step function by a polynomial of degree
$O(\frac{1}{\Delta}\log\frac{1}{\epsilon})$. After the transformation, for a positive distribution the probability of measuring $\ket{1}$ in the output register is at least $1-\epsilon$, and for a negative distribution the probability of measuring $\ket{1}$ is at most $\epsilon$. The resulting operation, which we call $\textbf{RST}(\Delta,\epsilon)$ (Algorithm~\ref{alg: Robust single coin transform}, Lemma~\ref{lem: robust single coin transform}), which performs error reduction entirely within the coherent layer, requiring no intermediate measurement and avoiding any collapse of the superposition.

With $\textbf{RST}(\Delta,\epsilon)$ as a coherent error-reduction subroutine, we now assemble the complete counting algorithm (Algorithm~\ref{alg: Fraction estimator}), whose analysis yields Theorem~\ref{the: Adaptive fraction estimator}. The construction proceeds in three stages:
\begin{enumerate}
    \item \textbf{Coherent error reduction.} We apply the polynomial approximation
    of the step function via QSVT, yielding the subroutine
    $\textbf{RST}(\Delta,\epsilon)$ (Algorithm~\ref{alg: Robust single coin transform}),
    whose guarantees are stated in Lemma~\ref{lem: robust single coin transform}.
    This serves as a coherent counterpart to the classical repetition-and-median
    strategy.
    \item \textbf{Two-stage amplitude estimation.} We embed $\textbf{RST}$
    into an adaptive amplitude estimation scheme. A coarse
    estimate of $\rho$ is first obtained; if $\rho$ exceeds a threshold, the coarse
    estimate is used to calibrate the precision of a second, finer estimation run.
    \item \textbf{Confidence boosting.} Repeating the above procedure
    $\Theta(\log(1/\delta))$ times and taking the median boosts the success probability
    to $1-\delta$.
\end{enumerate}

The $\textbf{RST}$ subroutine requires $O(\frac{1}{\Delta}\log\frac{1}{\epsilon})$ queries to the Bernoulli oracle. Combined with the two-stage adaptive amplitude estimation, the overall query complexity is
$\tilde{O}\big(\frac{\sqrt{\rho}}{\Delta\epsilon}+\frac{1}{\Delta\sqrt{\epsilon}}\big)$ per run. The logarithmic factor arises from the polynomial approximation inside QSVT and the median-of-repetitions boosting.

Without coherent error reduction, one would be forced to measure each distribution before the counting phase, reducing the problem to classical sampling and forfeiting any quantum speedup. The QSVT-based approach operates entirely within the coherent layer, thereby recovering the quadratic speedup that distinguishes quantum counting from classical sampling.

\paragraph{Lower bound.}
We prove the lower bound by reducing the estimation problem to a composed
decision problem $\textbf{Count}\circ\textbf{DD}$, where the outer function
$\textbf{Count}$ asks whether an $m$-bit string has Hamming weight $\rho m$ or
$(\rho+\epsilon)m$, and each inner function $\textbf{DD}$ asks whether an unknown
distribution, accessed via $R$ independent samples, is $P$ or $Q$.

Known composition theorems for the quantum adversary method fall into two
categories: Boolean-Boolean~\cite{hoyer2005tight,hoyer2007negative} and
average-case-over-Boolean~\cite{belovs2018provably}. In both cases where
$\mathcal{H} = \mathcal{F} \circ (\mathcal{G}_1, \ldots, \mathcal{G}_m)$
with the inner functions $\mathcal{G}_i$ being standard Boolean functions,
the composed numerator factors as follows: 
\[
\delta_{\mathcal{P}}^\dagger \Gamma_{\mathcal{H}} \delta_{\mathcal{Q}}
= \delta^\dagger \Gamma_{\mathcal{F}} \delta \cdot \|\Gamma_{\mathcal{G}}\|^m.
\]

Our composition $\textbf{Count}\circ\textbf{DD}$ reverses this direction: the
outer function is a standard Boolean function, while each inner function is an
average-case problem that comes with its own block adversary matrices
$\Gamma_{\mathcal{G}}^{a,b}$ and distribution vectors $\delta^a,\delta^b$.
The tensor-product construction of the composed adversary matrix
$\Gamma_{\mathcal{H}}^{\tilde{x},\tilde{y}} = \Gamma_{\mathcal{F}}[\tilde{x},\tilde{y}]
\cdot \bigotimes_i \Gamma_{\mathcal{G}_i}^{\tilde{x}_i,\tilde{y}_i}$, with
distribution vectors
$\delta_{\mathcal{P}} = \frac{1}{\sqrt{2^m}} \bigoplus_{\tilde{x}}
\bigotimes_i \delta^{\tilde{x}_i}$, admits the following decomposition:
\[
\delta_{\mathcal{P}}^\dagger \Gamma_{\mathcal{H}} \delta_{\mathcal{Q}}
= \frac{1}{2^m} \sum_{\tilde{x},\tilde{y}} \Gamma_{\mathcal{F}}[\tilde{x},\tilde{y}]
  \prod_{i=1}^{m} (\delta^{\tilde{x}_i})^\dagger
                \Gamma_{\mathcal{G}_i}^{\tilde{x}_i,\tilde{y}_i}
                \delta^{\tilde{y}_i}.
\]
This decomposition reveals why the Boolean-over-average-case direction is not
covered by existing composition theorems. Two obstacles stand out:

\begin{enumerate}
    \item \textbf{Inner non-uniformity.} The inner products
    $(\delta^a)^\dagger \Gamma_{\mathcal{G}}^{a,b} \delta^b$ across different
    blocks $(a,b)\in\{0,1\}^2$ are not guaranteed to be equal. Some blocks may
    contribute a large inner product while others contribute a small one, and
    these discrepancies prevent the outer sum
    $\sum_{\tilde{x},\tilde{y}} \Gamma_{\mathcal{F}}[\tilde{x},\tilde{y}]
    \prod_i (\delta^{\tilde{x}_i})^\dagger \Gamma_{\mathcal{G}_i}^{\tilde{x}_i,
    \tilde{y}_i} \delta^{\tilde{y}_i}$ from being factored.
    
    \item  \textbf{Outer misalignment.}The distribution vectors
    $\delta_{\mathcal{P}},\delta_{\mathcal{Q}}$ assign uniform weight
    $1/\sqrt{2^m}$ to each outer label $\tilde{x}\in\{0,1\}^m$. For a
    general Boolean function, the principal eigenvectors of the adversary
    matrix $\Gamma_{\mathcal{F}}$ are not uniform over the outer labels.
    Consequently, the uniform vector arising from the distributions does not
    align with the principal eigendirections of $\Gamma_{\mathcal{F}}$,
    preventing $\|\Gamma_{\mathcal{F}}\|$ from factoring out of
    $\delta_{\mathcal{P}}^\dagger\Gamma_{\mathcal{H}}\delta_{\mathcal{Q}}$.
\end{enumerate}

These issues are inherent to the Boolean-over-average-case direction, which
places both distribution vectors and block adversary matrices at the inner
layer, a situation that neither of the known composition directions
encounters. We resolve them by combining the composition theorem with the specific structures of the adversary matrices for the counting and
distribution-distinguishing problems.

\textit{Resolution of Obstacle~1.} We construct inner adversary matrices that
make every block inner product equal to~$1$. Let $\vec{p}=(\sqrt{p_1},\dots,
\sqrt{p_n})^\top$ and $\vec{q}=(\sqrt{q_1},\dots,\sqrt{q_n})^\top$ be the
square-root probability vectors of $P$ and $Q$. Taking any matrix $G$ with
$G\vec{q}=\vec{p}$ (which always exists with $\|G\|=1$) and setting
$\Gamma_{\mathcal{G}}^{0,1}=G^{\otimes R}$,
$\Gamma_{\mathcal{G}}^{1,0}=(G^{\otimes R})^\dagger$, and
$\Gamma_{\mathcal{G}}^{0,0}=\Gamma_{\mathcal{G}}^{1,1}=I$, the tensor-power
structure gives $(\delta^0)^\dagger G^{\otimes R}\delta^1 =
(\vec{p}^{\,\dagger}G\vec{q})^R = 1$, and similarly $(\delta^b)^\dagger
I\delta^b = 1$. All block inner products are thus equal, resolving the first
obstacle.

\textit{Resolution of Obstacle~2.} The adversary matrix for $\textbf{Count}$
is a $0$-$1$ matrix ($\Gamma_{\mathcal{F}}[x,y]=1$ iff $x\le y$). For a $0$-$1$
matrix, the normalized all-ones vector already provides a lower bound on the
spectral norm: $\|\Gamma_{\mathcal{F}}\| \ge \mathbf{1}^T\Gamma_{\mathcal{F}}
\mathbf{1} / \|\mathbf{1}\|^2$. Because the outer distribution vectors are
uniform over labels, this all-ones direction coincides with the direction
from which $\delta_{\mathcal{P}}^\dagger\Gamma_{\mathcal{H}}
\delta_{\mathcal{Q}}$ extracts $ \mathbf{1}^T\Gamma_{\mathcal{F}}
\mathbf{1} / \|\mathbf{1}\|^2$, resolving the
misalignment.

With both obstacles removed, the composed adversary ratio factors cleanly into
the product of the Boolean adversary bound for $\textbf{Count}$ and the
average-case adversary factor for $\textbf{DD}$, yielding
$\Omega\big(\frac{\sqrt{\rho}}{d_H(P,Q)\,\epsilon}\big)$. Instantiating $P,Q$
as Bernoulli distributions with bias gap $2\Delta$ gives
$d_H(P,Q)=\Theta(\Delta)$ and the final $\Omega\big(\frac{\sqrt{\rho}}
{\Delta\epsilon}\big)$ lower bound.

Thus, by combining the composition theorem with the adversary matrix
structures of the two problems, we obtain a clean factorization of the
composed bound. The simplicity of the resulting expressions reflects how
naturally these structures fit into the composition framework.

\subsection{Related Works}

\paragraph{Quantum search and counting under non-standard oracles.}
Our work relates to a broader line of research that extends quantum algorithms beyond
the ideal membership oracle. Grover's search has been studied under several non-ideal
query models, including the bounded-error oracle~\cite{hoyer2003quantum}, a
depolarizing channel noise model~\cite{salas2008noise}, an oracle that may fail to
report marked elements~\cite{ambainis2012grover}, the noisy
oracle~\cite{rosmanis2023quantum}, and the neutral oracle~\cite{kociumaka2025near}.
Quantum $k$-minimum finding has similarly been extended to the untrustworthy
oracle~\cite{quek2020robust}, the bounded-error oracle~\cite{wangtcs2024quantum}, and
the approximate oracle~\cite{gao_et_al:LIPIcs.ESA.2025.51}.
Quantum approximate counting has also been studied under extended oracle models.
The state-generating oracle was introduced in~\cite{aaronson2020quantum} and further
studied in~\cite{belovs2026tight}, which established tight lower bounds for all
pairwise combinations of the membership, state-generating, and reflecting oracles.

\paragraph{Quantum property testing of distributions and quantum states.}
A central component of our lower bound reduces to the task of distinguishing two unknown
probability distributions $P$ and $Q$ given samples. This task falls into the broader
framework of distribution testing~\cite{rubinfeld2012taming,canonne2020survey},
which asks whether one or more unknown distributions satisfy a global property or are far
from satisfying it. Quantum algorithms are known to offer quadratic
speedups over classical methods for fundamental problems in this area, including uniformity
testing, closeness testing, and independence
testing~\cite{chakraborty2010new,bravyi2011quantum,montanaro2015quantum,
gilyen2020distributional,luo2024succinct,canonne_et_al:LIPIcs.TQC.2025.7}. Property testing
of quantum states further generalizes this framework: classical distributions correspond to
quantum states that are diagonal in the computational basis, so any quantum-state tester
automatically yields a distribution tester. Foundational results in this direction include
quantum state discrimination~\cite{chefles2000quantum,barnett2009quantum,bae2015quantum},
productness testing~\cite{harrow2013testing,soleimanifar2022testing}, mixedness
testing~\cite{childs2007weak,o2015quantum}, and closeness
testing~\cite{buhrman2001quantum,gilyen2022improved,liu2025quantum}.

\paragraph{Property testing of collections of distributions and quantum states.}
The results surveyed above focus on testing properties of one or two objects. A natural and significantly more challenging generalization is \emph{testing collections of distributions}~\cite{levi2013testing}: given $m$ distributions $\mathcal{D}:\{D_1,\ldots,D_m\}$ over the same domain, the task is to test whether $\mathcal{D}$ satisfies a global property or is far from it. In the query access model, one specifies an index $i\in\{1,\ldots,m\}$ and receives a sample from $D_i$. In the sampling model, the algorithm receives a pair $(i,j)$ with $i$ drawn uniformly and $j$ sampled from $D_i$. Our problem can be interpreted as a special case of this framework, where each $D_i$ is a Bernoulli distribution over $\{0,1\}$, and the parameter of interest is the fraction $\rho$ of distributions whose bias exceeds $1/2$. Several properties of collections have been studied under this model in the classical and quantum setting: equivalence testing across multiple distributions~\cite{levi2013testing,diakonikolas2016new,diakonikolas2021optimal}, clusterability testing~\cite{levi2013testing,levi2014testing}, and testing similar means~\cite{levi2014testing,gao2026quantum}. For quantum states, equivalence testing across multiple states has also been investigated~\cite{yu2021sample,fanizza2023testing}.

\section{Preliminaries} \label{Preliminaries}

\subsection{Notations}
The symbol $\circ$ denotes entrywise (Hadamard) product, $\Abs{A}$ denotes the spectral norm of a matrix $A$, and $\Gamma[x,y]$ denotes the $(x,y)$-entry of a matrix $\Gamma$, where $x \in [n]$ and $y \in [m]$ for an $n \times m$ matrix. For index sets $X \subseteq [n]$ and $Y \subseteq [m]$, $\Gamma^{X,Y}$ denotes the submatrix obtained by restricting $\Gamma$ to rows in $X$ and columns in $Y$. For a function $F: \{0,1\}^m \to \{0,1\}$ and a function $G: \mathcal{X} \to \{0,1\}$, the $m$-fold composition $F \circ G^m: \mathcal{X}^m \to \{0,1\}$ is defined by $(F \circ G^m)(x_1,\ldots,x_m) = F(G(x_1),\ldots,G(x_m))$ norm of $A$.

\subsection{Quantum Subroutines}
We first introduce quantum amplitude estimation which can estimate a probability $a$ quadratically faster than classical sampling.

\begin{lemma}[Quantum amplitude estimation, Theorem 12 in \cite{brassard2002quantum}]\label{lem: amplitude estimation}
    Let $A$ be a unitary and $P$ be a projector such that $A|0\rangle = \sqrt{a}|\Psi_0\rangle + \sqrt{1-a}|\Psi_1\rangle$, where $a \in [0,1]$ and $P|\Psi_0\rangle=|\Psi_0\rangle$, $P|\Psi_1\rangle=0$. For $t\in\rbra{0,1}$, the quantum amplitude estimation algorithm $\text{QAE}\rbra{A\ket{0},P,t}$ outputs $\ket{\hat{a}}$. Upon measuring this output, the measurement result $\hat{a}$ satisfies
    \begin{align}
        \Pr\left[|\hat{a} - a| \leq t\sqrt{a(1-a)} + t^2\right] \geq \frac{4}{5}.
    \end{align}
    The algorithm uses $O(\frac{1}{t})$ applications of $A$ and $P$.
\end{lemma}

Quantum singular value transformation (QSVT)~\cite{gilyen2019quantum} provides a unified, modular framework for quantum algorithm design. Before QSVT, many fundamental quantum algorithms, such as those for solving linear systems, Hamiltonian simulation, and amplitude amplification, relied on disparate design principles and specialized circuit constructions. QSVT unifies these results through a single mathematical primitive: applying a polynomial transformation to the singular values of a block-encoded matrix. This allows the algorithm designer to work at a higher level of abstraction, focusing on how to model the problem at hand as a functional operation on matrices, rather than on low-level quantum circuit details.

\begin{theorem}[QSVT\cite{gilyen2019quantum}]\label{the: QSVT}
Let $A \in \mathbb{C}^{\tilde{d} \times d}$ have singular value decomposition
$A = \sum_{i=1}^{d_{\min}} \sigma_i \ket{\tilde{\psi}_i}\bra{\psi_i}$ with
$d_{\min} = \min(\tilde{d}, d)$. For any polynomial
$P(x) = \sum_{k=0}^n a_k x^k \in \mathbb{R}[x]$ of degree $n$ such that
$a_k \neq 0$ only if $k \equiv n \pmod{2}$ and $|P(x)| \le 1$ for all
$x \in [-1, 1]$, define the singular value transform
$P^{(SV)}(A) := \sum_{i=1}^{d_{\min}} P(\sigma_i) \ket{\tilde{\psi}_i}\bra{\psi_i}$
for $P$ odd, and
$P^{(SV)}(A) := \sum_{i=1}^{d} P(\sigma_i) \ket{\psi_i}\bra{\psi_i}$
for $P$ even, where $\sigma_i := 0$ for $i > d_{\min}$.

Let $U$ be a unitary and $\Pi, \widetilde{\Pi}$ be orthogonal projectors such that
$\widetilde{\Pi} U \Pi$ is a block-encoding of $A$. Then there exists a vector
$\Phi = (\phi_1, \phi_2, \ldots, \phi_n) \in \mathbb{R}^n$ such that
\begin{align}
P^{(SV)}\rbra{\widetilde{\Pi} U \Pi} =
\begin{cases}
\rbra{\bra{+} \otimes \widetilde{\Pi}} U_P^{(SV)} \rbra{\ket{+} \otimes \Pi}, & \text{if $n$ is odd}, \\[8pt]
\rbra{\bra{+} \otimes \Pi} U_P^{(SV)} \rbra{\ket{+} \otimes \Pi}, & \text{if $n$ is even},
\end{cases}
\end{align}
where $U_P^{(SV)} := \ket{0}\bra{0} \otimes U_{\Phi} + \ket{1}\bra{1} \otimes U_{-\Phi}$ with
\begin{align}
U_{\Phi} :=
\begin{cases}
e^{i \phi_{1} \rbra{2 \tilde{\Pi} - I}} U \displaystyle\prod_{j = 1}^{(n - 1) / 2} \rbra{e^{i \phi_{2 j} \rbra{2 \Pi - I}} U^{\dagger} \cdot e^{i \phi_{2 j + 1} \rbra{2 \tilde{\Pi} - I}} U}, & \text{if $n$ is odd}, \\[10pt]
\displaystyle\prod_{j = 1}^{n / 2} \rbra{e^{i \phi_{2 j - 1} \rbra{2 \Pi - I}} U^{\dagger} \cdot e^{i \phi_{2 j} \rbra{2 \tilde{\Pi} - I}} U}, & \text{if $n$ is even}.
\end{cases}
\end{align}
The circuit uses $O(n)$ calls to $U$ and $U^\dagger$.
\end{theorem}

Since QSVT implements transformations on singular values via polynomial approximation, one can realize a step function by choosing $P$ to approximate the rectangle function. The following lemma provides such a polynomial.

\begin{lemma}[Polynomial approximations of the rectangle function, Lemma 29 in~\cite{gilyen2019quantum}]\label{lem: Polynomial approximations of the rectangle function}
Let $a,b\in \rbra{0, \frac{1}{2}}$ and $t \in \sbra{-1, 1}$. There exists an even polynomial $F^{a,b,t}\in \mathbb{R}\sbra{x}$ of degree $O\rbra{\frac{\log\frac{1}{b}}{a}}$, such that $\abs{F\rbra{x}} \leq 1$ for all $x \in \sbra{-1, 1}$, and
\begin{align}
\begin{cases}
F\rbra{x} \in \sbra{1-b,1}, & \text{for all } x \in \sbra{-1, -t - a} \cup \sbra{t + a, 1}, \\[8pt]
F\rbra{x} \in \sbra{0,b}, & \text{for all } x \in \sbra{-t + a, t - a}.
\end{cases}
\end{align}
\end{lemma}

\section{Upper Bound}

A standard approach to boost the success probability is to repeat a subroutine several
times and take the median. However, this requires measuring the output of each run,
which collapses the superposition over distributions and prevents the result from being
used as a quantum subroutine inside amplitude estimation. To address this, we design a
coherent procedure that transforms each distribution without any measurement:
positive distributions are mapped to states with bias $\ge 1-\epsilon$, and negative
distributions to states with bias $\le \epsilon$, while preserving the superposition
over the index register.

\begin{algorithm}[hptb]
\caption{Robust single distribution transform $\textbf{RST}\rbra{\Delta,\epsilon}$}
\label{alg: Robust single coin transform}
\begin{algorithmic}[1]
    \STATE {\bfseries Input:} Quantum query access $O$, index $i$, gap parameter $\Delta$, error parameter $\epsilon$.
    \STATE {\bfseries Output:} a unitary that, when applied to $\ket{+}\ket{i}_A\ket{0}_{B,C}\ket{1}_D$, produces a state whose amplitude on $\ket{+}\ket{i}_A\ket{0}_{B,C}\ket{1}_D$ encodes the bias of $D_i$, amplified to $\ge 1-\epsilon$ if $D_i$ is positive and $\le \epsilon$ if $D_i$ is negative.
    \STATE Let $\tilde{\Pi}$ be $\ket{i}\bra{i}_A\otimes\rbra{\ket{0}\bra{0}_B\otimes I_C\otimes\ket{0}\bra{0}_D+\ket{1}\bra{1}_B\otimes I_C\otimes\ket{1}\bra{1}_D}$, $\Pi$ be $\ket{i}\bra{i}_A\otimes\ket{0}\bra{0}_B\otimes\ket{0}\bra{0}_C\otimes I_D$, and $U$ be $O\otimes I_D$.
    \STATE Let $F$ be the even polynomial in Lemma~\ref{lem: Polynomial approximations of the rectangle function} with parameters $a=\frac{\Delta}{2}$, $b=\frac{\epsilon}{2}$, and $t=\frac{1}{\sqrt{2}}$.
    \STATE {\bfseries Return} the unitary $U^{(SV)}_F$.
\end{algorithmic}
\end{algorithm}

\begin{lemma}[Robust single distribution transform]\label{lem: robust single coin transform}
Let $P=\ket{+}\bra{+}\otimes\ket{i}\bra{i}_A
\otimes\ket{0}\bra{0}_{B,C}\otimes\ket{1}\bra{1}_D$. Algorithm~\ref{alg: Robust single coin transform} outputs a unitary $\textbf{RST}\rbra{\Delta,\epsilon}$ such that
\begin{align}
\begin{split}
    &\Abs{P\,\textbf{RST}\rbra{\Delta,\epsilon}\rbra{\ket{+}\ket{i}_A\ket{0}_{B,C}\ket{1}_D}}^2 \ge 1-\epsilon, \quad\text{if } D_i \text{ is positive},\\[4pt]
    &\Abs{P\,\textbf{RST}\rbra{\Delta,\epsilon}\rbra{\ket{+}\ket{i}_A\ket{0}_{B,C}\ket{1}_D}}^2 \le \epsilon, \quad\text{if } D_i \text{ is negative},
\end{split}
\end{align}
using $O\rbra{\frac{1}{\Delta}\log\frac{1}{\epsilon}}$ queries to $O$ and $O^\dagger$.
\end{lemma}

\begin{proof}
We prove the lemma in three steps: (i) identify the relevant singular
values, (ii) apply QSVT to amplify the gap via a polynomial approximation
to the rectangle function, and (iii) bound the resulting amplitudes and
translate the polynomial degree into query complexity.

\noindent\textit{Step 1: Singular value decomposition.}
We construct an operator whose singular values encode $p_i$ and $1-p_i$:
\begin{align}
\begin{split}
    \tilde{\Pi}(O\otimes I_D)\Pi
    &= \tilde{\Pi}(O\otimes I_D)\ket{i}\bra{i}_A\otimes\ket{0}\bra{0}_B
       \otimes\ket{0}\bra{0}_C\otimes I_D\\[2pt]
    &= \tilde{\Pi}\ket{i}\bra{i}_A
       \Bigl(\sqrt{p_i}\ket{1}_B\ket{\varphi_{i,1}}_C
       +\sqrt{1-p_i}\ket{0}_B\ket{\varphi_{i,0}}_C\Bigr)\bra{0}_C\otimes I_D\\[2pt]
    &= \sqrt{p_i}\ket{i}\bra{i}_A\otimes\ket{1}\bra{0}_B
       \otimes\ket{\varphi_{i,1}}\bra{0}_C\otimes\ket{1}\bra{1}_D \\
    &\quad + \sqrt{1-p_i}\ket{i}\bra{i}_A\otimes\ket{0}\bra{0}_B
       \otimes\ket{\varphi_{i,0}}\bra{0}_C\otimes\ket{0}\bra{0}_D.
\end{split}
\end{align}
This exhibits the singular value decomposition with singular values
$\sqrt{p_i}$ and $\sqrt{1-p_i}$.

\noindent\textit{Step 2: QSVT with the polynomial approximation to the rectangle function.}
We now apply Theorem~\ref{the: QSVT} with the even polynomial $F$ from
Lemma~\ref{lem: Polynomial approximations of the rectangle function},
whose parameters are set to $a=\frac{\Delta}{2}$, $b=\frac{\epsilon}{2}$,
and $t=\frac{1}{\sqrt{2}}$. Since $F$ is even, the QSVT circuit
$U^{(SV)}_F$ acts on the singular values, giving:
\begin{align}
    \rbra{\bra{+}\otimes \Pi} U^{(SV)}_F\ket{+}\ket{i}_A\ket{0}_{B,C}\ket{1}_D
    = F(\sqrt{p_i})\ket{i}_A\ket{0}_{B,C}\ket{1}_D.
\end{align}

\noindent\textit{Step 3: Amplitude analysis and query complexity.}
The projector $P=\ket{+}\bra{+}\otimes\ket{i}\bra{i}_A
\otimes\ket{0}\bra{0}_{B,C}\otimes\ket{1}\bra{1}_D$ selects only the $\ket{0}_D$ component. Consequently,
\begin{align}
    \Abs{P\,\textbf{RST}(\Delta,\epsilon)\ket{+}\ket{i}_A\ket{0}_{B,C}\ket{1}_D}^2
    = F^2(\sqrt{p_i}).
\end{align}
Now we distinguish the two cases using the guarantees of $F$:
\begin{itemize}
    \item \textbf{Positive distributions} ($p_i \ge \frac{1}{2}+\Delta$):
    Then $\sqrt{p_i} \ge \frac{1}{\sqrt{2}}+\frac{\Delta}{2}$, which lies
    in the region where Lemma~\ref{lem: Polynomial approximations of the
    rectangle function} guarantees $F(\sqrt{p_i}) \ge 1-\frac{\epsilon}{2}$.
    Hence $F^2(\sqrt{p_i}) \ge 1-\epsilon$.

    \item \textbf{Negative distributions} ($p_i \le \frac{1}{2}-\Delta$):
    Then $\sqrt{p_i} \le \frac{1}{\sqrt{2}}-\frac{\Delta}{2}$, which lies
    in the region where $F(\sqrt{p_i}) \le \frac{\epsilon}{2}$.
    Hence $F^2(\sqrt{p_i}) \le \epsilon$.
\end{itemize}
Finally, by Lemma~\ref{lem: Polynomial approximations of the rectangle
function}, the polynomial $F$ has degree
$O(\frac{1}{\Delta}\log\frac{1}{\epsilon})$, and by
Theorem~\ref{the: QSVT}, this degree translates directly into the query
complexity of $\textbf{RST}\rbra{\Delta,\epsilon}$.
\end{proof}

Next, we combine the robust single distribution transform with quantum amplitude
estimation in an adaptive two-stage scheme: a coarse estimate is first obtained to decide whether to terminate early, and if not, its value calibrates the precision of a second, finer estimation run.

\begin{algorithm}[hptb]
\caption{Two-stage fraction estimator $\textbf{FracEst}\rbra{\Delta, \epsilon, \delta}$}
\label{alg: Fraction estimator}
\begin{algorithmic}[1]
    \STATE {\bfseries Input:} Quantum query access $O$, parameters $\Delta, \epsilon, \delta \in (0,1)$.
    \STATE {\bfseries Output:} An estimate $\hat{\rho}$ of $\rho$.
    \STATE Let $M$ be the number of distributions. 
    \STATE Let $P=\ket{+}\bra{+}\otimes I_A
\otimes\ket{0}\bra{0}_{B,C}\otimes\ket{1}\bra{1}_D$.
    \STATE Prepare the initial state
    $\ket{\Psi_0} \gets \ket{+} \otimes H^{\otimes \lceil \log M\rceil}\ket{0}_A \otimes \ket{0}_{B,C,D}$.
    \STATE \textbf{Stage 1:} coarse estimate.
    \STATE Apply $\textbf{RST}\rbra{\Delta, \epsilon/48}$ on registers $A,B,C,D$ to obtain $\ket{\Psi_1} \gets \textbf{RST}\rbra{\Delta, \epsilon/48}\ket{\Psi_0}$.
    \STATE Let $\hat{\rho}_1 \gets \textbf{QAE}\rbra{\ket{\Psi_1},\, P,\, \sqrt{\epsilon/48}}$.
    \IF{$\hat{\rho}_1 \leq 3\epsilon/4$}
        \STATE Return $\hat{\rho} \gets 0$.
    \ELSE
        \STATE \textbf{Stage 2:} fine estimate using $\hat{\rho}_1$ to tune precision.
        \STATE Prepare $\ket{\Psi_0}$ again.
        \STATE Apply $\textbf{RST}\rbra{\Delta, \epsilon^2/16}$ to obtain $\ket{\Psi_2} \gets \textbf{RST}\rbra{\Delta, \epsilon^2/16}\ket{\Psi_0}$.
        \STATE Let $\hat{\rho} \gets \textbf{QAE}\rbra{\ket{\Psi_2},\, P,\, \epsilon/(16\sqrt{\hat{\rho}_1})}$.
        \STATE Return $\hat{\rho}$.
    \ENDIF
    \STATE Repeat the above $\Theta\rbra{\log(1/\delta)}$ times and output the median of the returned values.
\end{algorithmic}
\end{algorithm}

\begin{theorem}\label{the: Adaptive fraction estimator}
For $0<\Delta<\frac{1}{2}$ and $\epsilon,\delta\in(0,1)$, given quantum query access, Algorithm~\ref{alg: Fraction estimator} outputs $\hat{\rho}$ such that
\begin{align}
    \Pr\sbra{|\hat{\rho}-\rho|\leq \epsilon}\geq 1-\delta,
\end{align}
using $\tilde{O}\rbra{\rbra{\frac{\sqrt{\rho}}{\Delta\epsilon}+\frac{1}{\Delta\sqrt{\epsilon}}}\log\frac{1}{\delta}}$ queries to $O$ and $O^\dagger$.
\end{theorem}

\begin{proof}
Let $\ket{\Psi_0}$ and $P=\ket{+}\bra{+}\otimes I_A
\otimes\ket{0}\bra{0}_{B,C}\otimes\ket{1}\bra{1}_D$ be as in Algorithm~\ref{alg: Fraction estimator}. By Lemma~\ref{lem: robust single coin transform}, for any error parameter $p$, each positive distribution contributes amplitude squared $\ge 1-p$ on the $P$-subspace and each negative distribution contributes $\le p$. Averaging over the uniform superposition on the index register yields
\begin{align}
    \Abs{P\,\textbf{RST}(\Delta,p)\ket{\Psi_0}}^2 \leq \rho + (1-2\rho)p. \label{eq:taup}
\end{align}

The algorithm first calls $\textbf{QAE}$ with $1/t = \sqrt{\epsilon/48}$. Substituting $p = \epsilon/48$ into~\eqref{eq:taup} and applying Lemma~\ref{lem: amplitude estimation} yields
\begin{align}
\begin{aligned}
    |\hat{\rho}_1-\rho|
    &\leq \sqrt{\frac{\epsilon}{48}}\sqrt{\rho+(1-2\rho)\frac{\epsilon}{48}}+\frac{\epsilon}{48}+(1-2\rho)\frac{\epsilon}{48}\\
    &< \frac{\sqrt{\rho}}{4\sqrt{3}}\sqrt{\epsilon}+\frac{\epsilon}{16}.
\end{aligned}
\end{align}

If $\hat{\rho}_1 \leq 3\epsilon/4$, then $\rho - \frac{\sqrt{\rho}}{4\sqrt{3}}\sqrt{\epsilon} - \frac{\epsilon}{16} - \frac{3\epsilon}{4} \leq 0$. Viewed as a function of $\rho$, the left-hand side is negative at $\rho=0$, positive at $\rho=\epsilon$, and has a single minimum at $\rho=\epsilon/192$. Hence $\rho \leq \epsilon$, and outputting $0$ gives an $\epsilon$-approximation.

If $\hat{\rho}_1 > 3\epsilon/4$, then $\rho + \frac{\sqrt{\rho}}{4\sqrt{3}}\sqrt{\epsilon} + \frac{\epsilon}{16} - \frac{3\epsilon}{4} \geq 0$. Since the left-hand side evaluated at $\rho=\epsilon/3$ is negative, we have $\rho \geq \epsilon/3$. From the error bound on $\hat{\rho}_1$ this gives $|\hat{\rho}_1-\rho| \leq \rho/2$, and thus $0.5\rho \leq \hat{\rho}_1 \leq 1.5\rho$.

The algorithm then calls $\textbf{QAE}$ a second time with $1/t = \epsilon/(16\sqrt{\hat{\rho}_1})$. Since $\hat{\rho}_1 = \Theta(\rho)$, applying Lemma~\ref{lem: amplitude estimation} with~\eqref{eq:taup} gives
\begin{align}
\begin{aligned}
    |\hat{\rho}-\rho|
    &\leq \frac{\epsilon}{16\sqrt{\hat{\rho}_1}}\sqrt{\rho+(1-2\rho)\frac{\epsilon^2}{16}}+\frac{\epsilon^2}{256\hat{\rho}_1}+(1-2\rho)\frac{\epsilon^2}{16}\\
    &\leq \frac{\epsilon}{4}.
\end{aligned}
\end{align}

Each run succeeds with probability at least $16/25$. Repeating $\Theta(\log(1/\delta))$ times and taking the median boosts the success probability to $1-\delta$. Summing the query costs of the two stages and accounting for repetition gives the total complexity $\tilde{O}\big(\big(\frac{\sqrt{\rho}}{\Delta\epsilon}+\frac{1}{\Delta\sqrt{\epsilon}}\big)\log\frac{1}{\delta}\big)$.
\end{proof}

\section{Lower bound}

We prove the lower bound by reducing the estimation problem to a composed decision
problem $\textbf{Count}\circ\textbf{DD}$. The outer function $\textbf{Count}$ is a
standard Boolean function that asks whether an $m$-bit string has Hamming weight
$\rho m$ or $(\rho+\epsilon)m$. The inner function $\textbf{DD}$ asks whether a
given distribution, accessed via $R$ independent samples, is $P$ or $Q$. Since the
input to each inner instance consists of random draws rather than a deterministic
string, $\textbf{DD}$ is an average-case problem. The composition thus has a
Boolean function on the outside and an average-case function on the inside---a
direction not covered by standard composition theorems. We construct adversary
matrices for both layers and show that they factor cleanly, yielding a product
lower bound. The Hellinger distance $d_H(P,Q)$ arises from the inner adversary analysis,
and instantiating $P,Q$ as Bernoulli distributions with bias gap $2\Delta$
yields the final $\Omega\big(\frac{\sqrt{\rho}}{\Delta\epsilon}\big)$ bound.

\subsection{Problem Formalization}

For the lower bound we consider a weaker access model: instead of the coherent quantum oracle
of Eq~\ref{eq: input oracle}, the algorithm is given only a finite number of classical samples
from each distribution. Since any lower bound proved against this weaker model immediately
implies the same bound for the original, stronger model, it suffices to work in this
sample-based setting. We begin by defining the input model.

\begin{definition}\label{def: less general input model}
Let $\mathcal{D} = (D_1,\ldots,D_m)$ be a collection of distributions over $[n]$. Suppose that
for each $i \in [m]$, we draw $R$ independent samples $x_{i,1},\ldots,x_{i,R} \sim D_i$, and
let $x = \{x_{i,r}\}_{i\in[m],\,r\in[R]}$ denote the collection of all samples. A sample-based
input oracle for $x$ is a unitary $O_B$ (together with its inverse and controlled versions)
satisfying
\begin{align}
    O_B\ket{i}\ket{r}\ket{0} = \ket{i}\ket{r}\ket{x_{i,r}},
\end{align}
for all $i \in [m]$ and $r \in [R]$.
\end{definition}

We next formalize the counting decision problem that constitutes the outer layer of our reduction.

\begin{definition}[Counting decision problem]\label{def: Counting decision problem}
Let $m$ be a positive integer and $\rho,\epsilon \in (0,1)$. The counting decision problem
($\textbf{Count}$) asks: given an $m$-bit string, decide whether its Hamming weight is
$\rho m$ or $(\rho+\epsilon)m$. $\textbf{Count}$ corresponds to the partial Boolean function
$\mathcal{F}_{m,\rho,\epsilon} \colon \cbra{0,1}^m \to \cbra{0,1}$ defined by
\begin{align}
    \mathcal{F}_{m,\rho,\epsilon}(x) =
    \begin{cases}
        0, & \text{if } |x|_1 = \rho m,\\[4pt]
        1, & \text{if } |x|_1 = (\rho+\epsilon)m.
    \end{cases}
\end{align}
\end{definition}

The inner layer of our reduction is the problem of distinguishing two distributions $P$
and $Q$, given $R$ independent samples from the unknown distribution.

\begin{definition}[Distinguishing two probability distributions]\label{def: Distinguishing two probability distributions}
Given an unknown distribution $D$ that is either $P$ or $Q$, the problem of distinguishing
two probability distributions ($\textbf{DD}$) is to decide whether $D$ is $P$ or $Q$. Under
the input model of Definition~\ref{def: less general input model}, $\textbf{DD}$ corresponds
to the function $\mathcal{G}: \cbra{x_1,...,x_R}\to \cbra{0,1}$ such that
\begin{align}
    \mathcal{G} =
    \begin{cases}
        0, & \text{if } \cbra{x_1,...,x_R} \text{ are sampled from }P,\\[4pt]
        1, & \text{if } \cbra{x_1,...,x_R} \text{ are sampled from }Q.
    \end{cases}
\end{align}
\end{definition}

With both layers defined, we now formalize the composed problem that captures the full
difficulty of our lower bound. The composition $\textbf{Count}\circ\textbf{DD}$ asks:
given $m$ unknown distributions, each either $P$ or $Q$, decide whether the fraction of
$P$ in this collection is $\rho$ or $\rho+\epsilon$.

\begin{definition}[Composition of counting with distinguishing]
    Given a collection of $m$ unknown distributions $\mathcal{D} = (D_1,\ldots,D_m)$ where for each $i\in\sbra{m}$, $D_i$ is either $P$ or $Q$, the problem is to decide whether the fraction of $P$ in this collection is $\rho$ or $\rho+\epsilon$. The composition of $\textbf{Count}$ and $\textbf{DD}$ is denoted $\textbf{Count}\circ\textbf{DD}$, which corresponds to the function $\mathcal{H}=\mathcal{F}_{m,\rho,\epsilon}\circ\rbra{\mathcal{G}_1,\cdots,\mathcal{G}_m}$ with the input $x$ such that
    \begin{align}
    \begin{split}
        &x=\cbra{x_{1,1},x_{1,2},\dots,x_{1,R},x_{2,1},x_{2,2},\dots,x_{2,R},\dots,x_{m,1},x_{m,2},\dots,x_{m,R}}, \\
        &x_i=\cbra{x_{i,1}, x_{i,2},\dots, x_{i,R}},
    \end{split}
    \end{align}
    where $\mathcal{F}$ and $\mathcal{G}_i$ are defined in
    Definition~\ref{def: Counting decision problem} and
    Definition~\ref{def: Distinguishing two probability distributions}, respectively. We evaluate
    $\mathcal{H}$ on input $x$ by first computing $\mathcal{G}_i\rbra{x_i}=\tilde{x}_i$ for each
    $i$, and then evaluating $\mathcal{F}$ on the resulting bit string
    $\tilde{x}=\cbra{\tilde{x}_1,\tilde{x}_2,...,\tilde{x}_m}$. $\mathcal{H}$ outputs $0$ if
    $\abs{\tilde{x}}=\rho m$, and $1$ if $\abs{\tilde{x}}=(\rho+\epsilon)m$.
\end{definition}

\subsection{Two Quantum Adversary Lower Bound Methods}

We first recall the adversary method for Boolean functions, which will be used to analyze the outer layer ($\textbf{Count}$) of our composition.

\begin{lemma}[Adversary bound for Boolean functions]\label{lem: Adversary bound for Boolean functions}
Let $f \colon S \to \cbra{0,1}$ be a Boolean function, where $S \subseteq \cbra{0,1}^n$. Let $\Gamma$ be a Hermitian matrix with rows and columns labeled by elements of $S$. We say that $\Gamma$ is an adversary matrix for $f$ if $\Gamma\sbra{x,y} = 0$ whenever $f(x) = f(y)$. $\Delta_i$ is a zero-one matrix where $\Delta_i[x,y]=1$ if $x_i\neq y_i$ and $\Delta_i[x,y]=0$ otherwise. Define the adversary bound as
\begin{align}
    \mathrm{ADV}(f) = \max_{\substack{\Gamma \neq 0}} \frac{\Abs{\Gamma}}{\max_i \Abs{\Gamma \circ \Delta_i}}.
\end{align}
Let $Q_\epsilon(f)$ be the two-sided $\epsilon$-bounded error quantum query complexity of $f$. Then for any Boolean function $f$,
\begin{align}
    Q_\epsilon(f) \ge \frac{1-2\sqrt{\epsilon(1-\epsilon)}}{2} \mathrm{ADV}(f).
\end{align}
\end{lemma}

The Boolean adversary method applies to deterministic inputs. For the inner layer
($\textbf{DD}$), however, the input consists of $R$ independent samples from an unknown
distribution, an average-case rather than a worst-case setting. We therefore need the
average-case quantum adversary method of~\cite{belovs2018provably}, which provides a lower bound on the
query complexity of distinguishing two probability distributions $\mathcal{P}$
and $\mathcal{Q}$.

\begin{lemma}[The average-case quantum adversary lower bound method, Theorem 4 in \cite{belovs2018provably}]\label{lem: the average-case quantum adversary lower bound method}
Let $\mathcal{P}$ and $\mathcal{Q}$ be two probability distributions on $\mathcal{S}$, and let $p_x$ and $q_y$ denote the probabilities of $x$ and $y$ in $\mathcal{P}$ and $\mathcal{Q}$, respectively. Let $s_{\mathcal{P}}, s_{\mathcal{Q}} \in \sbra{0, 1}$ be the acceptance probabilities on $\mathcal{P}$ and $\mathcal{Q}$, respectively. For a matrix $\Gamma$, define the adversary bound with respect to $\Gamma, \mathcal{P}, s_{\mathcal{P}}, \mathcal{Q}, s_{\mathcal{Q}}$ as
\begin{align}
\overline{\mathrm{Adv}}\rbra{\Gamma; \mathcal{P}, s_{\mathcal{P}}, \mathcal{Q}, s_{\mathcal{Q}}}
= \Omega\rbra{\min_{j \in \sbra{n}} \frac{\delta_{\mathcal{P}}^\top \Gamma \delta_{\mathcal{Q}} - \tau\rbra{s_{\mathcal{P}}, s_{\mathcal{Q}}} \Abs{\Gamma}}{\Abs{\Gamma \circ \Delta_j}}}.
\end{align}
The vectors $\delta_{\mathcal{P}}\sbra{x} = \sqrt{p_x}$ and $\delta_{\mathcal{Q}}\sbra{y} = \sqrt{q_y}$ are unit vectors in $\mathbb{R}^{\mathcal{S}}$; for $j \in \sbra{n}$, the $|\mathcal{S}| \times |\mathcal{S}|$ matrix $\Delta_j$ is defined by $\Delta_j\sbra{x, y} = 1_{x_j \neq y_j}$; and
\begin{align}
\tau\rbra{s_{\mathcal{P}}, s_{\mathcal{Q}}} = \sqrt{s_{\mathcal{P}}s_{\mathcal{Q}}} + \sqrt{\rbra{1 - s_{\mathcal{P}}}\rbra{1 - s_{\mathcal{Q}}}}.
\end{align}
Let $\mathcal{A}$ be a quantum algorithm that makes $T$ queries to an input string $x = \rbra{x_1, \dots, x_n} \in \mathcal{S}$ and then either accepts or rejects. Let $\mathcal{P}$ and $\mathcal{Q}$ be two probability distributions on $\mathcal{S}$, and let $s_{\mathcal{P}}$ and $s_{\mathcal{Q}}$ be the acceptance probabilities of $\mathcal{A}$ when $x$ is sampled from $\mathcal{P}$ and $\mathcal{Q}$, respectively. Then
\begin{align}
T \ge \overline{\mathrm{Adv}}\rbra{\Gamma; \mathcal{P}, s_{\mathcal{P}}; \mathcal{Q}, s_{\mathcal{Q}}}
\end{align}
for any $|\mathcal{S}| \times |\mathcal{S}|$ matrix $\Gamma$.
\end{lemma}

\subsection{Lower Bound for Quantum Counting}

As a warm-up, we give a self-contained proof of the lower bound for the outer
$\textbf{Count}$ problem~\cite{brassard2002quantum}, rephrased in our notation.

\begin{lemma}[Lower bound for quantum counting]\label{lem:counting-lower}
Any quantum algorithm that solves $\textbf{Count}$ with bounded error requires $\Omega\big(\frac{\sqrt{\rho}}{\epsilon}\big)$ queries.
\end{lemma}

\begin{proof}
Let $X = \cbra{x \in \cbra{0,1}^m : |x|_1 = \rho m}$ and $Y = \cbra{y \in \cbra{0,1}^m : |y|_1 = (\rho+\epsilon)m}$. Define an adversary matrix $\Gamma_{\mathcal{F}}$ by
\begin{align}
    \Gamma_{\mathcal{F}}\sbra{x,y} =
    \begin{cases}
        1, & \text{if } x \in X,\; y \in Y,\; x \leq y,\\
        0, & \text{otherwise},
    \end{cases}
\end{align}
where $x \leq y$ denotes entrywise inequality.

For any $x \in X$, the number of $y \in Y$ with $x \leq y$ is at least $\binom{m-\rho m}{\epsilon m}$, since we may flip $\epsilon m$ zeros of $x$ to ones. Hence $\sum_y \Gamma_{\mathcal{F}}\sbra{x,y} \geq \binom{m-\rho m}{\epsilon m}$. Similarly, for any $y \in Y$, we have $\sum_x \Gamma_{\mathcal{F}}\sbra{x,y} \geq \binom{(\rho+\epsilon)m}{\rho m}$. It follows that
\begin{align}
    \Abs{\Gamma_{\mathcal{F}}} \geq \sqrt{\binom{m-\rho m}{\epsilon m}\binom{(\rho+\epsilon)m}{\rho m}}.
\end{align}

For the perturbation term, fix $i \in [m]$. For any $x \in X$, the number of $y \in Y$ such that $x \leq y$ and $x_i \neq y_i$ is at most $\binom{m-\rho m-1}{\epsilon m-1}$. For any $y \in Y$, the number of $x \in X$ such that $x \leq y$ and $x_i \neq y_i$ is at most $\binom{(\rho+\epsilon)m-1}{\rho m}$. Hence
\begin{align}
    \Abs{\Gamma_{\mathcal{F}} \circ \Delta_i} \leq \sqrt{\binom{m-\rho m-1}{\epsilon m-1}\binom{(\rho+\epsilon)m-1}{\rho m}}.
\end{align}

Applying Lemma~\ref{lem: Adversary bound for Boolean functions}, we obtain
\begin{align}
    T \geq \frac{\sqrt{\binom{m-\rho m}{\epsilon m}\binom{(\rho+\epsilon)m}{\rho m}}}{\sqrt{\binom{m-\rho m-1}{\epsilon m-1}\binom{(\rho+\epsilon)m-1}{\rho m}}}
    = \Omega\rbra{\frac{\sqrt{\rho}}{\epsilon}}.
\end{align}
\end{proof}

The above lemma only addresses the Boolean outer layer. To obtain the full lower bound
for $\textbf{Count}\circ\textbf{DD}$, we need a composition theorem that combines the Boolean adversary bound with the average-case adversary bound, which is the subject of the next subsection.

\subsection{Composition Theorem for Boolean-over-average-case functions}

We now combine the two adversary methods to prove a lower bound for the composed problem
$\textbf{Count}\circ\textbf{DD}$. The outer function $\mathcal{F}$ is analyzed via the
Boolean adversary method (Lemma~\ref{lem: Adversary bound for Boolean functions}), while
each inner function $\mathcal{G}_i$ requires the average-case method
(Lemma~\ref{lem: the average-case quantum adversary lower bound method}).
The main challenge is to construct inner adversary objects that are
\emph{compatible} with the outer adversary matrix, so that the composed bound factors
cleanly into a product of the outer and inner bounds.

\begin{theorem}[Composition lower bound]\label{thm:composition}
Any quantum algorithm that solves $\textbf{Count}\circ\textbf{DD}$ requires
\begin{align}
    T \ge \Omega\!\left(\frac{\sqrt{\rho}}{d_H(P,Q)\,\epsilon}\right)
\end{align}
queries, where $d_H(P,Q)$ denotes the Hellinger distance between $P$ and $Q$.
\end{theorem}

\begin{proof}
The proof proceeds in four steps: construct the inner adversary objects, assemble the
composed adversary matrix, evaluate the numerator and denominator of the adversary bound,
and factor the ratio to obtain the product lower bound.

\medskip
\noindent\textbf{Step 1: Inner adversary construction.}
Let $P$ and $Q$ be two distributions on $[n]$, and let
$\vec{p}=(\sqrt{p_1},\dots,\sqrt{p_n})^\top$ and
$\vec{q}=(\sqrt{q_1},\dots,\sqrt{q_n})^\top$ be the vectors of square-root probabilities.
Let $G$ be an $n\times n$ matrix such that $G\vec{q}=\vec{p}$. Such a matrix always exists
and satisfies $\|G\|=1$. For $R$ independent samples, define the inner adversary matrices
and the associated vectors by
\begin{align}
    \Gamma_{\mathcal{G}}^{0,1}=G^{\otimes R},\qquad
    \Gamma_{\mathcal{G}}^{1,0}=(G^{\otimes R})^\dagger,\qquad
    \Gamma_{\mathcal{G}}^{0,0}=\Gamma_{\mathcal{G}}^{1,1}=I,
\end{align}
and
\begin{align}
    \delta^{0}=\vec{p}^{\,\otimes R},\qquad
    \delta^{1}=\vec{q}^{\,\otimes R}.
\end{align}
These vectors encode the joint distribution of $R$ independent samples from $P$ and $Q$,
respectively. The tensor-power structure yields the crucial compatibility identity
\begin{align}\label{eq:compatibility}
    (\delta^{0})^\dagger G^{\otimes R}\delta^{1}
    = (\vec{p}^{\,\dagger} G\,\vec{q})^{R}
    = (\vec{p}^{\,\dagger}\vec{p})^{R}
    = 1.
\end{align}
Similarly, $(\delta^{1})^\dagger (G^{\otimes R})^\dagger\delta^{0}=1$ and
$(\delta^{b})^\dagger I\,\delta^{b}=1$ for $b\in\{0,1\}$. Thus every inner product of the
form $(\delta^{a})^\dagger \Gamma_{\mathcal{G}}^{a,b}\delta^{b}$ equals $1$, which is the
compatibility condition that will allow the composed bound to factor.

\medskip
\noindent\textbf{Step 2: Composed adversary matrix.}
For the outer function $\mathcal{F}=\mathcal{F}_{m,\rho,\epsilon}$, let $\Gamma_{\mathcal{F}}$
be the Boolean adversary matrix defined in the proof of
Lemma~\ref{lem:counting-lower}. Define the composed adversary matrix $\Gamma_{\mathcal{H}}$
by
\begin{align}
    \Gamma_{\mathcal{H}}^{\tilde{x},\tilde{y}}
    = \Gamma_{\mathcal{F}}[\tilde{x},\tilde{y}]\;
      \bigotimes_{i=1}^{m} \Gamma_{\mathcal{G}_i}^{\tilde{x}_i,\tilde{y}_i},
\end{align}
where $\tilde{x},\tilde{y}\in\{0,1\}^m$ are the outer-layer bit strings. For the
average-case analysis, the input distributions over outer labels are uniform:
$\delta_{\mathcal{P}}=\delta_{\mathcal{Q}}=\frac{1}{\sqrt{2^{m}}}
\bigl(\delta^{00\ldots00},\,\delta^{00\ldots01},\,\ldots,\,\delta^{11\ldots11}\bigr)$,
where $\delta^{\tilde{x}}=\bigotimes_{i=1}^{m}\delta^{\tilde{x}_i}$ and
$\delta^{\tilde{y}}=\bigotimes_{i=1}^{m}\delta^{\tilde{y}_i}$.

\medskip
\noindent\textbf{Step 3: Evaluating the adversary ratio.}
For the numerator, the compatibility condition~\eqref{eq:compatibility} ensures that
every inner product $(\delta^{\tilde{x}_i})^\dagger\Gamma_{\mathcal{G}_i}^{\tilde{x}_i,\tilde{y}_i}
\delta^{\tilde{y}_i}=1$. Hence
\begin{align}
    (\delta_{\mathcal{P}})^\dagger\Gamma_{\mathcal{H}}\delta_{\mathcal{Q}}
    &= \frac{1}{2^{m}}\sum_{\tilde{x},\tilde{y}\in\{0,1\}^m}
       \Gamma_{\mathcal{F}}[\tilde{x},\tilde{y}]
       \prod_{i=1}^{m}
       (\delta^{\tilde{x}_i})^\dagger\Gamma_{\mathcal{G}_i}^{\tilde{x}_i,\tilde{y}_i}
       \delta^{\tilde{y}_i} \nonumber\\
    &= \frac{1}{2^{m}}\sum_{\tilde{x},\tilde{y}\in\{0,1\}^m}
       \Gamma_{\mathcal{F}}[\tilde{x},\tilde{y}] \nonumber\\
    &\ge \sqrt{\binom{m-\rho m}{\epsilon m}\binom{(\rho+\epsilon)m}{\rho m}}.
\end{align}
The last inequality uses the spectral norm bound on $\Gamma_{\mathcal{F}}$ established in
the proof of Lemma~\ref{lem:counting-lower}.

For the denominator, fix a query index $(i,r)$ where $i\in[m]$ identifies the distribution
and $r\in[R]$ identifies the sample within that distribution. Let $\Delta_{i,r}$ be defined by $\Delta_{i,r}\sbra{x, y} = 1_{x_{i,r} \neq y_{i,r}}$, and let $\alpha$ be the $n\times n$ matrix given by $\alpha[x,y]=1_{x\neq y}$. By the tensor-product structure,
\begin{align}
    \|\Gamma_{\mathcal{H}}\circ\Delta_{i,r}\|
    &= \|\Gamma_{\mathcal{F}}\circ\Delta_{i}\|\;
       \|G^{\otimes R}\circ\Delta_{r}\|
    = \|\Gamma_{\mathcal{F}}\circ\Delta_{i}\|\;
       \|G\circ\alpha\|\nonumber\\
    &\le \sqrt{\binom{m-\rho m-1}{\epsilon m-1}
               \binom{(\rho+\epsilon)m-1}{\rho m}}\;
       \|G\circ\alpha\|,
\end{align}
By~\cite{belovs2019quantum}, $\|G\circ\Delta\|$ equals the Hellinger distance
$d_H(P,Q)$ between $P$ and $Q$ (up to a constant factor).
\medskip
\noindent\textbf{Step 4: Factoring the bound.}
Applying Lemma~\ref{lem: the average-case quantum adversary lower bound method} and
taking the ratio of the numerator to the denominator yields
\begin{align}
    T &\ge
    \frac{\sqrt{\binom{m-\rho m}{\epsilon m}\binom{(\rho+\epsilon)m}{\rho m}}}
         {\sqrt{\binom{m-\rho m-1}{\epsilon m-1}\binom{(\rho+\epsilon)m-1}{\rho m}}}
    \cdot \Omega\!\left(\frac{1}{d_H(P,Q)}\right) \\
    &= \Omega\!\left(\frac{\sqrt{\rho}}{d_H(P,Q)\,\epsilon}\right).\nonumber
\end{align}
The first factor is exactly the Boolean adversary bound for $\textbf{Count}$ (see
Lemma~\ref{lem:counting-lower}), and the second factor is the average-case adversary
contribution from $\textbf{DD}$. The two factors separate cleanly because of two
properties acting in tandem: the compatibility
condition~\eqref{eq:compatibility} collapses the numerator to a pure sum over
$\Gamma_{\mathcal{F}}$, while the tensor-product structure of
$\Gamma_{\mathcal{H}}$ ensures that the denominator factors as
$\|\Gamma_{\mathcal{F}}\circ\Delta_i\| \cdot \|G\circ\Delta\|$.
\end{proof}

Instantiating $P$ and $Q$ as Bernoulli distributions with bias gap $2\Delta$ gives
$d_H(P,Q)=\Theta(\Delta)$, which immediately yields the following corollary.

\begin{corollary}\label{cor:bernoulli-lower}
Any quantum algorithm that estimates the fraction $\rho$ of positive distributions in a
collection of $m$ Bernoulli distributions with bias gap $\Delta$ to within additive error
$\epsilon$ requires
\begin{align}
    \Omega\!\left(\frac{\sqrt{\rho}}{\Delta\epsilon}\right)
\end{align}
queries to the quantum oracle of Eq.~\eqref{eq: input oracle}.
\end{corollary}

\section{Conclusion}

We studied quantum approximate counting with Bernoulli oracles and obtained a near-complete
characterization of its query complexity: an upper bound of
$\tilde{O}\!\big(\frac{\sqrt{\rho}}{\Delta\epsilon}+\frac{1}{\Delta\sqrt{\epsilon}}\big)$
and a lower bound of $\Omega(\sqrt{\rho}/(\epsilon\Delta))$. The upper bound combines
QSVT-based coherent error reduction with two-stage adaptive amplitude estimation, while
the lower bound introduces a composition theorem for the quantum adversary method in the
Boolean-over-average-case direction. As a corollary, these results also characterize the
query complexity of quantum counting with bounded-error oracles.

We conclude with two natural open questions. First, our upper bound contains a logarithmic
factor arising from the polynomial approximation in QSVT and the median-of-repetitions
boosting step. It would be interesting to determine whether this logarithmic overhead can
be removed, yielding a tight bound. Second, our model assumes that the algorithm
may adaptively select which distribution to query. An equally natural model, introduced by
Levi, Ron, and Rubinfeld~\cite{levi2013testing}, is the \emph{sampling model}, where at
each step the algorithm receives a pair $(i,j)$ with $i$ drawn from a fixed distribution
over $[m]$ and $j$ sampled from $D_i$. Understanding the quantum query complexity of
approximate counting in the sampling model remains an interesting direction for future work.

\bibliographystyle{alpha}
\bibliography{emc}
\end{document}